\PassOptionsToPackage{usenames,svgnames}{xcolor}

\documentclass[a4, 11pt, parskip=half, DIV=11]{scrartcl}
\usepackage[english]{babel}
\usepackage[T1]{fontenc}
\usepackage[utf8]{inputenc}

\usepackage{graphicx}
\usepackage[numbers]{natbib}
\usepackage{subcaption}
\usepackage{booktabs}
\usepackage{url}
\usepackage{todonotes}
\usepackage{multirow}
\usepackage{csquotes}
\usepackage{paralist}

\usepackage{pgfplots}
\usepgfplotslibrary{colorbrewer}

\pgfdeclarelayer{background}
\pgfsetlayers{background,main}

\usepackage{bm}

\usepackage{amsmath}
\usepackage{amssymb}
\usepackage{amsthm}
\usepackage{amsfonts}
\usepackage{thmtools}		
\usepackage{mleftright}
\usepackage{stmaryrd}
\usepackage{nicefrac}

\theoremstyle{definition}
\newtheorem{theorem}{Theorem}

\newtheorem{definition}[theorem]{Definition}

\usepackage{thm-restate}
\usepackage[mathic=true]{mathtools}
\usepackage{fixmath}
\usepackage{siunitx}
\usepackage{color}

\usepackage{pifont}

\usepackage{enumitem}
\setlist[enumerate]{itemsep=0.2ex, topsep=0.5\topsep}
\setlist[description]{itemsep=0.2ex, topsep=0.5\topsep}
\setlist[itemize]{itemsep=0.2ex, topsep=0.5\topsep}

\usepackage{setspace}
\usepackage{ellipsis}
\usepackage{xspace}
\usepackage{hfoldsty}

\usepackage{tcolorbox}
\usepackage{lmodern}
\usepackage[tt=false]{libertine}
\usepackage[varqu]{zi4}
\usepackage[libertine]{newtxmath}

\usepackage{graphicx}
\usepackage{amsmath}
\usepackage{amsthm}
\usepackage{booktabs}
\usepackage{algorithm}

\usepackage[
pdfa,
hidelinks,
pdftex, 
pdfdisplaydoctitle,
pdfpagelabels,
pdfauthor={},
pdftitle={},
pdfsubject={},
pdfkeywords={},
pdfproducer={Latex with the hyperref package},
pdfcreator={pdflatex}
]{hyperref}
\usepackage{doi}
\usepackage{abstract}

\makeatletter
\def\thmt@refnamewithcomma #1#2#3,#4,#5\@nil{%
	\@xa\def\csname\thmt@envname #1utorefname\endcsname{#3}%
	\ifcsname #2refname\endcsname
	\csname #2refname\expandafter\endcsname\expandafter{\thmt@envname}{#3}{#4}%
	\fi
}
\makeatother
\usepackage[capitalise,noabbrev]{cleveref}   
\usepackage{mathtools}
\DeclareMathOperator{\depth}{depth}
\DeclareMathOperator{\children}{chi}
\DeclareMathOperator{\ged}{GED}
\DeclareMathOperator{\sdm}{SDM}
\DeclareMathOperator{\sdted}{SDTED}
\DeclareMathOperator{\diam}{diam}

\newcommand{\bbN}{\ensuremath{\mathbb{N}}}

\newcommand{\C}{\ensuremath{\bm{C}}}

\usepackage[switch]{lineno}
\usepackage{subcaption}
\usepackage{algpseudocode}

\algnewcommand\algorithmicforeach{\textbf{for each}}
\algdef{S}[FOR]{ForEach}[1]{\algorithmicforeach\ #1\ \algorithmicdo}
\algnewcommand{\algorithmicand}{\textbf{ and }}
\algnewcommand{\algorithmicor}{\textbf{ or }}
\algnewcommand{\OR}{\algorithmicor}
\algnewcommand{\AND}{\algorithmicand}
\algnewcommand\algorithmicnot{\textbf{not}}
\algdef{SE}[IF]{IfNot}{EndIf}[1]{\algorithmicif\ \algorithmicnot\ #1\ \algorithmicthen}{\algorithmicend\ \algorithmicif}%
\let\oldReturn\Return
\renewcommand{\Return}{\State\oldReturn}
\algtext*{EndWhile}%
\algtext*{EndIf}%
\algtext*{EndFor}%
\algtext*{EndFunction}%

\usepackage{lipsum}

\usepackage[auth-lg]{authblk}

\recalctypearea

\begin{document}
\title{Approximating the Graph Edit Distance with Compact Neighborhood Representations}

\author[$\dag$,1,2]{Franka Bause}
\author[$\dag$,3]{Christian Permann}
\author[1,4]{Nils M.~Kriege}

\affil[1]{University of Vienna, Faculty of Computer Science, Vienna, Austria}
\affil[2]{University of Vienna, UniVie Doctoral School Computer Science, Vienna, Austria}
\affil[3]{University of Vienna, Department of Pharmaceutical Sciences, Vienna, Austria}
\affil[4]{University of Vienna, Research Network Data Science, Vienna, Austria}
\date{\vspace{-30pt}}

\maketitle
\def\thefootnote{$\dag$}\footnotetext{Authors contributed equally to this work.}
\begin{abstract}
\unboldmath
The graph edit distance is used for comparing graphs in various domains. Due to its high computational complexity it is primarily approximated.
Widely-used heuristics search for an optimal assignment of vertices based on the distance between local substructures. 
While faster ones only consider vertices and their incident edges, leading to poor accuracy, other approaches require computationally intense exact distance computations between subgraphs.
Our new method abstracts local substructures to neighborhood trees and compares them using efficient tree matching techniques. This results in a ground distance for mapping vertices that yields high quality approximations of the graph edit distance. By limiting the maximum tree height, our method supports steering between more accurate results and faster execution.
We thoroughly analyze the running time of the tree matching method and propose several techniques to accelerate computation in practice. We use compressed tree representations, recognize redundancies by tree canonization and exploit them via caching.
Experimentally we show that our method provides a significantly improved trade-off between running time and approximation quality compared to existing state-of-the-art approaches.
\end{abstract}

\section{Introduction}
\label{sec:introduction}

In recent years graphs have been widely used to model structured data and methods for analyzing them have gained in popularity. Graphs are commonly used in web content mining~\cite{web_content}, pattern recognition~\cite{pattern_recognition}, molecular property prediction~\cite{doi:10.1021/ci050039t,doi:10.1021/acs.jcim.8b00820}, and many other domains. Various tasks and applications require quantifying the similarity of two graphs. One of the most prominent measures of graph similarity is the \emph{graph edit distance} (GED), i.e., the cost for transforming one graph into another by applying predefined edit operations. Usually adding, removing and relabeling vertices and edges is allowed, where each operation is assigned a specific cost. The computation of the graph edit distance from $G$ to $H$ then consists of finding a sequence of edit operations of minimum total cost that can be applied to $G$ and transforms it into a graph $G'$ isomorphic to $H$.
In practice, the usefulness of the exact GED is limited, as its computation is  $\mathsf{NP}$-hard~\cite{2_ComparingStars}. This has lead to the development of fast heuristic algorithms, such as \emph{bipartite graph matching} (BGM)~\cite{27_Riesen}. This method derives an edit path from an optimal (linear) assignment between the vertices of the two graphs. The cost for assigning a vertex $v$ in $V(G)$ to a vertex $v'$ in $V(H)$ reflects the cost for editing $v$ and its incident edges to match $v'$ and its incident edges. Vertex insertion and deletion is encoded by introducing additional placeholder vertices.
The derived edit path is not necessarily optimal, but its cost is used as an approximation of the GED in many domains~\cite{survey_BGM}. However, there are applications, for which the approximation quality obtained using this method is not sufficient, or that could benefit from more accurate results, such as graph clustering, $k$-nn classification or similarity search in graph databases. Hence, extensions of BGM have been proposed, which take enlarged local substructures around vertices into account. However, these either capture only limited structural information using walks~\cite{bagsofwalks} or require exact GED computation between subgraphs~\cite{10.1007/978-3-319-18224-719} leading to a drastic increase in both memory and time complexity.
For this reason, we propose a new GED heuristic following the concept of BGM, which includes local information in the form of \emph{neighborhood trees} to drastically improve the trade-off between running time and approximation quality. We compare neighborhood trees by a \emph{structure and depth preserving tree edit distance} (SDTED), which is highly efficient compared to computing the exact GED, while retaining the benefits of incorporating expressive structural information.

The concept of neighborhood trees is inspired by so-called \emph{unfolding trees} related to the Weisfeiler-Leman heuristic for the graph isomorphism problem. The method iteratively refines vertex labels encoding vertex neighborhoods of increasing size. The technique has recently become popular in graph learning, where it forms the basis of successful graph kernels~\cite{JMLR:v12:shervashidze11a,DBLP:journals/ans/KriegeJM20} and is fundamental to the expressivity analysis of graph neural networks~\cite{DBLP:journals/corr/abs-1810-00826,DBLP:journals/corr/abs-1810-02244,wl_survey}.
The discrete colors representing unfolding trees have been used to obtain a highly efficient heuristic for the GED ignoring their subtle structural differences at the expense of approximation quality~\cite{OALin}. Also, graph neural networks have been employed to solve the regression task of predicting the GED for graph pairs, but the technique does not allow to obtain a corresponding edit path~\cite{simgnn}.
Unfolding trees are a suitable concept for understanding the structure encoded by iterative methods based on direct neighborhoods, but are rarely generated explicitly. A main reason for this is that they quickly grow in size with increasing height.
Our concept of neighborhood trees overcomes this problem by pruning repeated vertices and allowing compact representations without redundancy amenable to efficient tree matching.

\begin{figure*}[tb]
	\includegraphics[width=\textwidth]{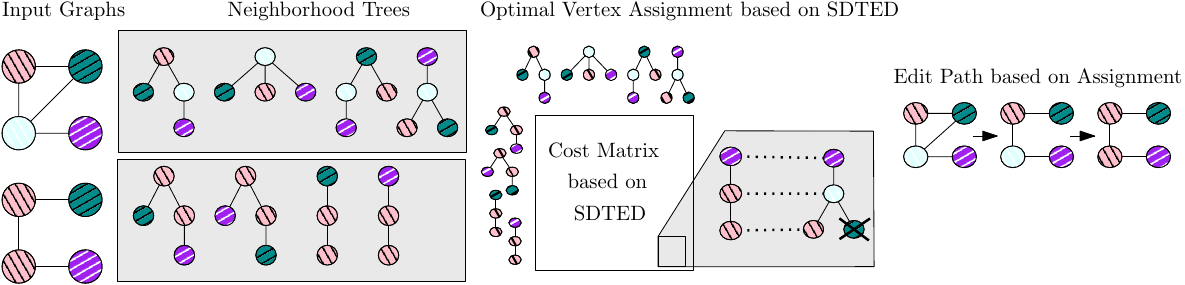}
	\caption{Overview of our approach. To approximate the graph edit distance between two graphs, an optimal assignment between their vertices is computed based on a cost matrix. Each entry of the cost matrix is obtained from SDTED between the neighborhood trees of the associated vertices. After an optimal assignment is found, an edit path between the two graphs is derived from it and the cost of this (suboptimal) edit path is the approximated graph edit distance.}
	\label{fig:overview}
\end{figure*}

\noindent
\textbf{Our contribution.} We make the following contributions.
\begin{enumerate}%
	\item We propose to use tree structures encoding node neighborhoods and their SDTED in the BGM framework for approximating the graph edit distance. An overview of our approach is depicted in Figure~\ref{fig:overview}.
	\item We propose \emph{neighborhood trees}, a compact and powerful variation of the Weisfeiler-Leman unfolding trees.
	\item We thoroughly analyze the running time of the SDTED and improve it to $O(n^2\Delta)$ for two trees with $n$ vertices and maximum degree $\Delta$ reducing it by a factor of $\Delta^2h$ over the best previously known result~\cite{80_generalizedWLkernel}, where $h$ is the height of the tree.
	\item We apply caching techniques and compressed representations to drastically speed up the computation in practice.
	\item We show that our new approach outperforms state-of-the-art approximation algorithms for the GED regarding approximation quality while being computed efficiently.
\end{enumerate}

\section{Preliminaries}
\label{sec:preliminaries}
In this section, we give an overview of the necessary definitions and the notation used throughout the article. We first introduce graphs, the graph edit distance and a standard approximation scheme for it. Then we provide an introduction to the Weisfeiler-Leman algorithm and the SDTED.

\subsection{Graph Theory}
	A \emph{graph} $G = (V,E,\mu,\nu)$ consists of a set of vertices $V$, a set of edges $E \subseteq V\times V$ between them and labeling functions $\mu \colon V \rightarrow L$ and $\nu\colon  E \rightarrow L$ for the vertices and edges, respectively. 
In this article, we assume that $L$ is a set of categorical labels. We consider undirected graphs and refer to an edge between $u$ and $v$ by $uv=vu$.
The vertices and edges of a graph $G$ are denoted by $V(G)$ and $E(G)$, respectively.
The \emph{neighbors} of a vertex $u \in V$ are denoted by $N(u) = \{v\mid  uv\in E\}$.
A \emph{path} is a sequence of vertices $(v_0, \dots, v_n)$ such that 
$v_{i}v_{i+1} \in E$ for $0 \leq i < n$. The vertices and the edges 
connecting consecutive vertices are said to be \emph{contained} in the path. 
The \emph{length} of a path is the number of edges it contains.
A \emph{shortest path} between two vertices $u,v \in V$ is a path $(u,\dots, v)$ of minimum length.
The \emph{diameter} of a graph $G$ is the maximum shortest-path length between any two vertices in $G$ and denoted by $\diam(G)$.
	A \emph{rooted tree} $T$ is a connected, acyclic graph with a distinct vertex $r \in V(T)$ referred to as \emph{root} denoted by $r(T)$.
	For $v \in V(T)\backslash \{r\}$ the \emph{parent} $p(v)$ is defined as the unique vertex $u \in N(v)$ closer to $r$ than $v$, and $\forall v \in V(T)$ the \emph{children} are defined as $\children(v)= N(v)\backslash \{p(v)\}$.
	A node $v$ is a \emph{leaf} if $\children(v)=\emptyset$. We denote the set of leaves of a tree $T$ by $L(T)= \{v \in V(T) \mid \children(v)=\emptyset\}$.

	An \emph{isomorphism} between two graphs $G$ and $H$ is a bijection $\Phi\colon V(G)\rightarrow V(H)$ such that 
	\begin{enumerate}[label=(\roman*)]
	 \item $\forall u,v \in V(G)\colon \mu(v) = \mu(\Phi(v))$, and
	 \item $\forall u,v \in V(G)\colon \Phi(u)\Phi(v)\in E(H) \Leftrightarrow uv \in E(G)$, and
	 \item $\forall uv \in E(G)\colon  \nu(uv) = \nu(\Phi(u)\Phi(v))$.
	\end{enumerate}
	For two rooted trees $T$ and $T'$ in addition $\Phi(r(T))=r(T')$ must hold.
	If an isomorphism between two graphs $G$ and $H$ exists, they are called \emph{isomorphic}, denoted by $G \cong H$.

\subsection{Graph Edit Distance}
The graph edit distance (GED) is the cost of transforming one graph into the other by deletion, insertion and relabeling of vertices and edges.
An \emph{edit path} between $G$ and $H$ is a sequence $(e_1,e_2,\dots,e_k)$ of such edit operations that, when applied one after the other to $G$, yield a graph $G'\cong H$.
Edit operations have non-negative costs given by the edit cost function $c$. 
The GED is the cost of a cheapest edit path.
Formally, the \emph{graph edit distance} between two graphs $G$ and $H$ is defined as
\begin{linenomath*}
\begin{equation*}
 \ged(G,H) = \min\left\{ \sum_{i=1}^{k} c(e_i) \mathrel{\Big|} (e_1, \text{...}, e_k) \in \Upsilon(G,H) \right\},
\end{equation*}
\end{linenomath*}
 where $\Upsilon(G,H)$ denotes the set of all possible edit paths from $G$ to $H$.
Since computing the GED is $\mathsf{NP}$-hard~\cite{2_ComparingStars}, it is often approximated. Many heuristics rely on finding an optimal assignment between the vertices of the graphs.
	\label{defAssProb}
	Let $A$ and $B$ be two sets with $|A|=|B|=n$ and $c\colon A \times B \to \mathbb{R}$ a ground cost function. 
	An \emph{assignment} between $A$ and $B$ is a bijection $\pi\colon A \to B$. The \emph{cost} of an assignment $\pi$ is $c(\pi) = \sum_{a\in A} c(a,\pi(a))$.
	The \emph{assignment problem} is to find an assignment with minimum cost 
and  can be solved in cubic time using the Hungarian method~\cite{Munkres,Jonker2005ASA}.

\subsection{Bipartite Graph Matching}\label{sec:preliminaries:bgm}
An upper bound for the GED between two graphs can be obtained from any edit path between them. An edit path can be derived from an optimal assignment of their vertices regarding a cost matrix.
Riesen and Bunke~\cite{27_Riesen} constructed a $k\times k$ cost matrix $\C$, with $k=n+m$ for two graphs with $n$ and $m$ vertices, respectively. For simplicity, we present a reduced matrix with $k=\max\{n,m\}$ that is sufficient for metric edit costs~\cite{Serratosa15}. Due to the symmetry of the GED in this case no vertices need to be inserted. Assume w.l.o.g.~that $n\geq m$, then we obtain a quadratic $n \times n$ cost matrix by adding $n-m$ columns representing deletion of vertices in $G$ by defining
	\[
	\C=\left[ \begin{array}{cccc|ccc}
	c_{1,1} & c_{1,2} & \dots & c_{1,m} & c_{1,\epsilon} & \dots & c_{1,\epsilon} \\
	c_{2,1} & c_{2,2} & \dots & c_{2,m} & c_{2,\epsilon} & \dots & \vdots \\
	\vdots & \vdots & \ddots & \vdots & \vdots & \ddots  & \vdots\\
	c_{n,1} & c_{n,2} & \dots & c_{n,m} &  c_{n,\epsilon} & \dots & c_{n,\epsilon}\\
	\end{array}\right].
	\]
The entries of $\C$ can be interpreted as follows:
Values $c_{i,j}$ (left part) represent the costs for replacing vertex $i$ of $G$ with vertex $j$ of $H$.
The entries $c_{i,\epsilon}$ (right part) are the costs for deleting vertex $i$ of $G$ and its incident edges.
Based on this cost matrix an optimal assignment is computed. %
An equivalent solution can be obtained from an $n\times m$ cost matrix by a more involved reduction applicable even for non-metric edit costs~\cite{BougleuxGBB20}.

An edit path can be derived from a vertex assignment $\pi\colon V(G)\to V(H)$ as follows~\cite{27_Riesen}:
Any vertex
	\begin{enumerate}[label=(\roman*)]
		\item $v \in V(G)$ with $\pi(v) \notin V(H)$ is deleted,
		\item $v \in V(H)$ with $\pi^{-1}(v) \notin V(G)$ is inserted,
		\item $v \in V(G)$ with $\pi(v) \in V(H)$ and $\mu(v) \neq \mu(\pi(v))$ is relabeled.
	\end{enumerate}
Similarly any edge
	\begin{enumerate}[label=(\roman*)]
		\item $uv \in E(G)$ with $\pi(u)\pi(v) \notin E(H)$ is deleted,
		\item $uv \in E(H)$ with $\pi^{-1}(u)\pi^{-1}(v) \notin E(G)$ is inserted,
		\item $uv \in E(G)$ with $\pi(u)\pi(v) \in E(H)$ and $\nu(uv) \neq \nu(\pi(u)\pi(v))$ is relabeled.
	\end{enumerate}
The sum of edit costs is the approximated GED. Since the edit path found might not be optimal, this method provides an upper bound.

\subsection{Weisfeiler-Leman Color Refinement}
\label{pre:wlcr}
The 1-dimensional Weisfeiler-Leman algorithm or \emph{color refinement} starts with all vertices having a color representing their label (or a uniform coloring in case of unlabeled vertices). 
In each iteration the color of every vertex $v$ in $V(G)$ is updated based on the multiset of colors of its neighbors according to
\begin{linenomath*}
\begin{equation*}
c_{i+1}(v) = h\left(c_i(v),\{\!\!\{c_i(u)\mid u \in N(v)\}\!\!\}\right),
\end{equation*}
\end{linenomath*}
where $h$ is an injective function.
Since every graph isomorphism associates vertices having the same color only, the approach 
allows to recognize graphs as non-isomorphic.
Non-isomorphic graphs with the same multiset of colors for all $i\in \bbN$ exist, but they are rare even for $i=2$~\cite{Bab+1979} and many real-world graph datasets do not contain such graph pairs~\cite{DBLP:conf/ijcai/0001FK21}.
The expressivity and computational efficiency of the method makes it an important tool for graph isomorphism testing and machine learning~\cite{wl_survey,DBLP:conf/ijcai/0001FK21}.

\begin{figure}[tb]
	\centering
	\begin{subfigure}{0.15\linewidth}
		\includegraphics[height=0.1\textheight]{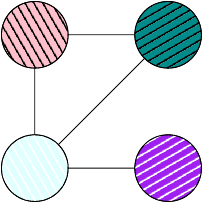}
	\end{subfigure}
	\begin{subfigure}{0.23\linewidth}
		\includegraphics[height=0.1\textheight]{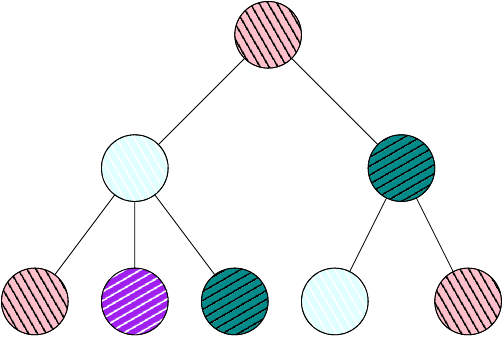}
	\end{subfigure}
	\begin{subfigure}{0.12\linewidth}
		\includegraphics[height=0.1\textheight]{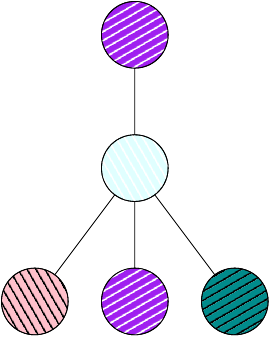}
	\end{subfigure}
	\begin{subfigure}{0.23\linewidth}
		\includegraphics[height=0.1\textheight]{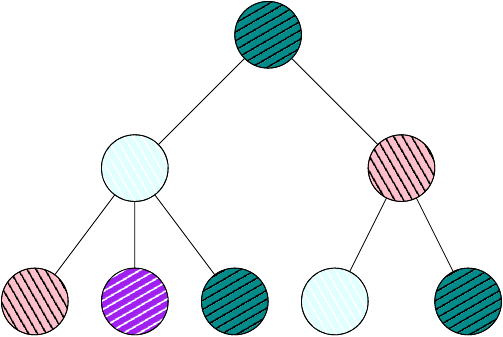}
	\end{subfigure}
	\begin{subfigure}{0.19\linewidth}
		\includegraphics[height=0.1\textheight]{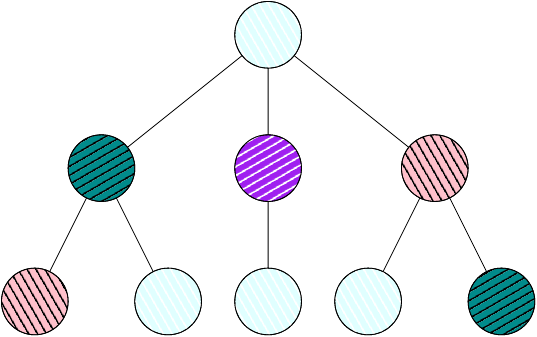}
	\end{subfigure}
	
	\caption{Graph $G$ and its unfolding trees $F_2^v, v\in V(G)$.}
	\label{fig:unfoldingtree}
\end{figure}

The colors encode neighborhoods of increasing radius by a tree structure, see Figure~\ref{fig:unfoldingtree}.
Let $\phi\colon V(T)\to V(G)$ be a mapping such that $\phi(n)=v$ if the node $n$ in $V(T)$ represents the vertex $v$ in $V(G)$.
Mathematically, the \emph{unfolding tree} $F_i^v$ with height $i$ of the vertex $v \in V(G)$ is defined recursively
	as the tree with a root $n_v$ with $\phi(n_v)=v$ and child subtrees $F^w_{i-1}$ for all $w \in N(v)$, and $F_0^v = (\{n_v\}, \emptyset)$. The labels of the original graph are preserved.
The unfolding trees $F_i^u$ and $F_i^v$ of two vertices $u$ and $v$ are isomorphic if and only if $c_i(u)=c_i(v)$.

\subsection{Structure and Depth Preserving Tree Edit Distance}\label{pre:sdted}
A suitable distance for (labeled) rooted trees is the \emph{structure and depth preserving tree edit distance} (SDTED)~\cite{80_generalizedWLkernel}, which is defined as
\begin{linenomath*}
\begin{equation*}
\sdted(T,T') := \min\{c'(M)\mid M \in \sdm(T,T')\},
\end{equation*}
\end{linenomath*}
where $\sdm(T,T')$ denotes the set of all structure and depth preserving mappings between $T$ and $T'$.
	A mapping $M \subseteq V(T) \times V(T')$ is \emph{structure and depth preserving}, iff
	\begin{enumerate}[label=(\roman*)]
		\item $\forall(u,u'),(v, v') \in M\colon$ $u = v \Leftrightarrow u'= v'$,
		\item $(r(T),r(T')) \in M$, and 
		\item $\forall(u,u') \in M\colon$ $(p(u),p(u'))\in M$. 
	\end{enumerate}
	Let $M_1=V(T)\setminus\{m_1\mid (m_1,m_2)\in M\}$ and $M_2= V(T')\setminus\{m_2\mid (m_1,m_2)\in M\}$ be the vertices not in the mapping.
	The cost $c'$ of the mapping $M$ is defined as
\begin{linenomath*}
\begin{equation*}
	c'(M)= \hspace{-.7em} \underbrace{\sum_{(u,v) \in M} \hspace{-.7em} c(\mu(u),\mu(v))}_{\text{Relabeling}} + 
	\underbrace{\sum_{u \in  M_1} c(u, \epsilon)}_{\text{Deletion}} + \underbrace{\sum_{v \in M_2} c(\epsilon,v)}_{\text{Insertion}}.
\end{equation*}
\end{linenomath*}
We can respect the label of an edge $p(v)v$ by associating the cost for its relabeling with the endpoint $v$.
Given two trees, the SDTED can be computed in a bottom-up fashion from optimal assignments between the children of two vertices at the same level until reaching the roots, see  Algorithm~\ref{alg:assign_simple} for a naive implementation of the SDTED. This simple variant does not include edge costs, but only vertex costs. Trees and subtrees with different numbers of children are padded with placeholder vertices such that deletion/insertion costs can be considered. The deletion cost of a whole subtree can be determined without additional recursion.  For our implementation of the SDTED we included edge costs and store the costs for all computed pairs of trees and subtrees using a lexicographic encoding, see Section~\ref{sec:caching}.
\begin{algorithm}[tb]
	\caption{Computation of SDTED}\label{alg:assign_simple}
	\begin{algorithmic}
		\Function{sdted}{Tree $T_1$, Tree $T_2$}
		\State $n \gets \max(\lvert \children(r(T_1))\rvert,\lvert \children(r(T_2))\rvert)$
		\State $ \textsc{pad}(T_1,n)$, $ \textsc{pad}(T_2,n)$\Comment{add undef. for equal size}
		\State $\C \gets \mathbb{R}^{n\times n}$\Comment{$0$ initialized}
		\ForEach {$s_i \in \children(r(T_1))$} 
		\ForEach {$s'_j \in \children(r(T_2))$}
		\If{$s_i$ is defined \OR $s'_j$ is defined }
		\If{$s_i$ is undefined}
		\State $c_{ij} \gets \textsc{deletionCost}(s'_j)$
		\ElsIf{$s'_j$ is undefined}
		\State $c_{ij} \gets \textsc{insertionCost}(s_i)$
		\Else 
		\State $c_{ij} \gets \textsc{sdted}(s_i,s'_j)$ \Comment{recursive call}
		\EndIf
		\EndIf
		\EndFor
		\EndFor
		\State $c_r \gets \textsc{cost}(r(T_1), r(T_2))$
		\State $c_s \gets \textsc{lapCost}(\C)$ \Comment{min. subtree assignment cost}
		\Return$c_r + c_s$
		\EndFunction
		
	\end{algorithmic}
\end{algorithm}

\section{Related Work}
\label{sec:relatedwork}

We discuss prior work on computing and approximating the GED as well as graph similarities based on the Weisfeiler-Leman algorithm.

\subsection{Exact Methods for Computing the GED}
 Computing the GED is a combinatorial optimization problem and different standard approaches for such problems exist and have been adapted specifically for the GED. 
 Some exact methods are based on backtracking search, e.g., A*-GED, DF-GED~\cite{10.5220/0005209202710278},
 CSI-GED~\cite{Gouda2016CSIGEDAE} or BSS\_GED~\cite{83_bssged}, while
 others use integer linear programming formulations such as BIP-GED~\cite{LerougeARHA15} or MIP-GED~\cite{BLUMENTHAL202046}.
 The complexity, however, makes exact approaches infeasible for large graphs and approximate methods are widely-used~\cite{survey_BGM}.

\subsection{Approximate Methods for the GED}
Many approximation methods~\cite{2_ComparingStars,27_Riesen,92_improved,29_ged_heuristics,10.1007/978-3-319-18224-719} are based on optimal assignments, which means they compute an assignment between the vertices of the graphs to find an upper or lower bound for the GED.
Usually, faster methods are only suitable for graphs with discrete labels and uniform edit costs~\cite{29_ged_heuristics}.
Riesen and Bunke~\cite{27_Riesen} introduced the basic concept of bipartite graph matching (BGM), where the costs for the assignment problem are estimated based on the vertices assigned to each other and their incident edges.
This very general approach can handle continuous labels and has later been improved to yield a tight lower bound for arbitrary (metric) edit costs, referred to as \textsc{Branch}~\cite{92_improved}.

While restricting to vertices and their incident edges allows to obtain lower bounds for the GED, better approximations can be obtained using larger substructures.
Carletti et al.~\cite{10.1007/978-3-319-18224-719} proposed to use neighborhood subgraphs, which are defined for a vertex $v$ of $G$ as the subgraph induced in $G$ by the vertices with a shortest-path distance from $v$ of at most $h$, where $h$ is a radius parameter. The cost for assigning a node $v$ to a node $v'$ are obtained by comparing their neighborhood subgraphs using the exact GED with the additional constraint that $v$ is assigned to $v'$. A downside of the approach is that the exact GED computation leads to a running time exponential in the size of the subgraphs, drastically increasing the computational complexity of the method.
Gaüzère et al.~\cite{bagsofwalks} associate so-called bags of walks $B_v$ with each vertex $v$, which contain all label sequences $(\mu(v_0),\nu(v_0v_1),\mu(v_1),\dots,\mu(v_h))$ associated with walks $(v_0,v_1,\dots,v_h)$ of length $h$ starting at the vertex $v=v_0$.
The assignment costs of a pair of vertices $v$ in $V(G)$ and $v'$ in $V(H)$ are obtained by comparing $B_v$ and $B_{v'}$. To avoid enumerating an exponential number of walks, the assignment cost is estimated from the $h$th powers of the adjacency matrices $\bm{A}_G$, $\bm{A}_H$ and $\bm{A}_{G\times H}$ of $G$, $H$ and their direct product graph $G\times H$, respectively. 
This comes at the cost of losing local structure information, as only the labels of the start and end vertices are considered and intermediate vertices are omitted. Moreover, the approximation quality decreases for walks with five or more vertices due to the phenomenon of \emph{tottering}~\cite{bagsofwalks}. In Table~\ref{tab:complexity} we compare the theoretical complexity of the methods based on BGM closely related to our approach.

\begin{table}[tb]
	\centering
	\caption{Time complexity of computing the cost matrix $\C$ for commonly used approximations for the graph edit distance based on the bipartite graph matching method. Here, $\Delta$ denotes the maximum degree, $h$ is the radius of the subgraphs/walk length and $\omega$ is the exponent of matrix multiplication. The best theoretic algorithms achieve $\omega<2.38$ but typically $\omega\approx2.81$ in practice~\cite{29_ged_heuristics}. In addition, we consider the complexity, when $\Delta$ is bounded by a constant.}
	\label{tab:complexity}
	\begin{tabular}{ccc}
		\toprule
		\textbf{Method} & \textbf{General graphs} & \textbf{Bounded-degree graphs} \\
		\midrule
		\emph{BGM}~\cite{27_Riesen} & $O(\vert V \vert^2 \Delta^3)$ & $O(\vert V \vert^2)$\\
		\emph{Walk}~\cite{bagsofwalks} & $O(\log(h)\vert V \vert^{2\omega})$& $O(\log(h)\vert V \vert^{2\omega})$ \\
		\emph{Subgraph}~\cite{10.1007/978-3-319-18224-719} & $O(\vert V \vert^2 \exp(\Delta^h))$& $O(\vert V \vert^2 \exp(h))$\\
		\textbf{\emph{Ours}} [Theorem~\ref{thm:complexity}] & $O(\vert V \vert^2 \vert E \vert^2 \Delta)$& $O(\vert V \vert^4)$\\
		\bottomrule
	\end{tabular}
\end{table}

GED approximation methods not based on BGM are often multiple orders of magnitude slower, while improving accuracy only marginally~\cite{29_ged_heuristics}.
This includes methods such as the LP relaxations of ILP formulations \cite{LerougeARHA15}, which may yield tighter lower bounds but, again, at drastically increased runtime.

\subsection{Weisfeiler-Leman based Methods for GED Approximation and Graph Similarity}
Our idea of neighborhood trees is inspired by Weisfeiler-Leman unfolding trees.
Graph and vertex similarities based on the Weisfeiler-Leman algorithm have been studied extensively in machine learning~\cite{DBLP:conf/ijcai/0001FK21}. Several graph kernels count matching  Weisfeiler-Leman colors~\cite{DBLP:journals/ans/KriegeJM20}, use them to compute optimal vertex assignments~\cite{Kriege2016b} or define similarities between colors, e.g., by matching unfolding trees~\cite{80_generalizedWLkernel}.
A GED approximation based on finding an optimal vertex assignment regarding a tree metric generated by the Weisfeiler-Leman algorithm was proposed by Kriege et al.~\cite{OALin}. This approach uses matching Weisfeiler-Leman colors and implicitly compares the unfolding trees of two vertices level-wise (from root to leaves) and determines the similarity as the number of levels up to which the unfolding trees are isomorphic.
This leads to a coarse similarity, in which vertices with different labels, but equal neighborhoods are considered less similar  than vertices with the same label and completely different neighborhoods. While the technique is faster than comparable BGM methods, it is less accurate on some datasets.

 Several graphs kernels are based on Weisfeiler-Leman colors~\cite{DBLP:journals/ans/KriegeJM20}, some of which measure vertex similarity by the number of iterations required until the vertices obtain different colors~\cite{Kriege2016b,TogninalliGLRB19}.
 The use of more fine-grained similarities between colors was first proposed in~\cite{Yanardag2015a}.
 Recently, this was realized by applying the SDTED to Weisfeiler-Leman unfolding trees~\cite{80_generalizedWLkernel}.
 A kernel based on the Wasserstein distance and sets of unfolding trees compared by an approximate tree edit distance was proposed in~\cite{WWLS22}. 
 The method is only feasible for unfolding trees of small height due to its high runtime and does not approximate the GED.
 
 Graph neural networks have been used for predicting the GED for pairs of graphs based on their vertex- and graph-level embeddings~\cite{simgnn}.
 This technique, however, does not yield an edit path and cannot guarantee that the result is a lower or upper bound for the GED. Also, datasets annotated with ground-truth GED are rare.
 The efficiency of graph neural networks and kernels operating on the unfolding tree (instead of the color) naturally depends on the tree size. Replacing unfolding trees by a more compact representation without losing structural information promises advantages in all these areas.

\subsection{Discussion}
Methods for approximating the GED face the trade-off between approximation quality and running time. The class of algorithms following the framework of BGM is particularly promising in the respect~\cite{29_ged_heuristics}.
However, in some applications the approximation quality of standard BGM algorithms is not sufficient.
Current approaches using larger substructures~\cite{10.1007/978-3-319-18224-719,bagsofwalks} rely on computationally expensive cost functions. Although this can potentially negate the advantage of BGM in terms of running time, the time complexity of these methods has often not been thoroughly analysed.
The method based on subgraphs~\cite{10.1007/978-3-319-18224-719} computes the exact GED between subgraphs of size $\Delta^h$, where $\Delta$ is the maximum degree and $h$ the radius parameter. Since exact GED algorithms require exponential time for each of the $\vert V \vert^2$ subgraph pairs, we obtain a total running time of $O(\vert V \vert^2\exp(\Delta^h))$ for computing the cost matrix.
For the method based on walks~\cite{bagsofwalks}, the $h$th power of the adjacency matrices of both graphs and their product graph is computed. Using exponentiation by squaring $O(\log(h))$ matrix multiplications are needed, where the matrix size is $n^2 \times n^2$ for the product graph of two graphs with $n$ vertices. From this all entries of the cost matrix can be derived, leading to a total running time of $O(\log(h)\vert V \vert^{2\omega})$, with $\omega$ the exponent of matrix multiplication.
We have summarized the results in Table~\ref{tab:complexity} and provide simplified running times for the case of input graphs with bounded-degree. This graph class is relevant as the GED is often applied to compare graphs with bounded degree such as molecular graphs.

In terms of time complexity, our approach is less efficient then standard BGM, but more efficient or competitive to existing approaches using larger substructures. In Section~\ref{sec:experiments} we show, that our method outperforms these state-of-the-art approximation methods for the GED in terms of approximation quality, while being computed efficiently in practice. Thereby, our method provides a favorable trade-off between efficiency and approximation quality.

\section{Neighborhood Trees for Graph Matching}
\label{sec:mainpart}

Accurate approximation of the GED through BGM requires representations of sufficiently large and complex vertex neighborhoods, which allow the efficient computation of suitable and expressive distances.
We propose to use trees for representing neighborhoods and use the SDTED between such trees as a cost function for the assignment, see Figure~\ref{fig:overview} for an overview of our approach.
Trees capture the vertex neighborhoods well and their structure can be exploited to speed up distance computation. This leads to significant speed-ups over neighborhood subgraphs. We investigate traditional Weisfeiler-Leman unfolding trees for this purpose, which quickly grow in size with increasing height and store a large amount of redundant information. To mitigate this, we propose \emph{$k$-redundant neighborhood trees}, which are equivalent to Weisfeiler-Leman unfolding trees for sufficiently large $k$ but exhibit less redundancy for smaller values of $k$ leading to significantly more compact representations. 
Of particular interest are the $0$-redundant neighborhood trees, which represent a vertex with its surroundings by a rooted tree of shortest paths up to a maximum length restricting the radius of the surroundings. 
Bounding the tree height allows to control the trade-off between more accurate results and faster execution.

\begin{figure}[]
	\centering
	\begin{subfigure}{0.14\linewidth}
		\centering
		\includegraphics[height=0.1\textheight]{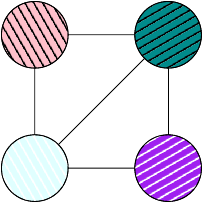}
		\subcaption{$G$}
	\end{subfigure}
	\begin{subfigure}{0.3\linewidth}
		\centering
		\includegraphics[height=0.1\textheight]{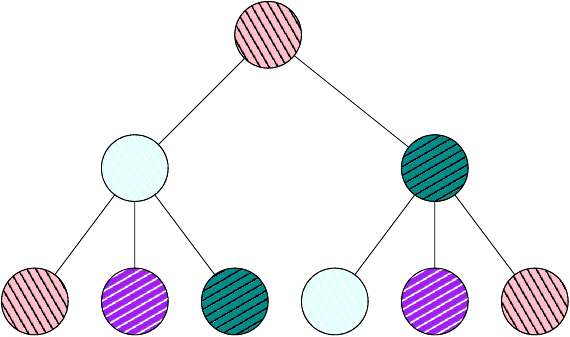}
		\subcaption{Unfolding tree}
	\end{subfigure}
	\begin{subfigure}{0.15\linewidth}
		\centering
		\includegraphics[height=0.1\textheight]{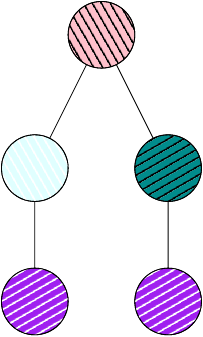}
		\subcaption{$0$-NT}
	\end{subfigure}
	\begin{subfigure}{0.15\linewidth}
		\centering
		\includegraphics[height=0.1\textheight]{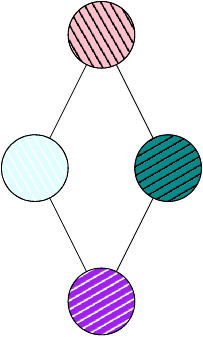}
		\subcaption{cNT}
	\end{subfigure}
\begin{subfigure}{0.2\linewidth}
	\centering
	\includegraphics[height=0.1\textheight]{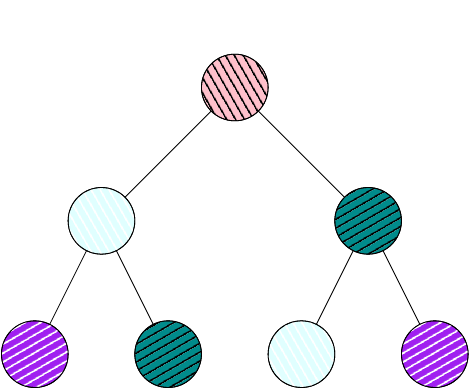}
	\subcaption{$1$-NT}
\end{subfigure}
	\caption{A graph $G$ and different tree representations with height $2$ of the neighborhood of the upper left vertex.}
	\label{fig:exampletrees}
\end{figure}

\subsection{Definitions and Properties}
We formalize the concept of neighborhood trees.
A $k$-redundant neighborhood tree can be constructed from the corresponding unfolding tree by deleting all the nodes that already occurred more than $k$ levels before together with the subtrees rooted at them.
Let $\depth(v)$ denote the length of the path from $v$ to the root and $\phi(v)$ denote the original vertex in $V(G)$ represented by $v$ in the unfolding or neighborhood tree.
\begin{definition}[$k$-redundant Neighborhood Tree, $k$-NT]\label{def:nt}
	For $k\geq 0$, the \emph{$k$-redundant neighborhood tree} of a vertex $v \in V(G)$ with height $i$, denoted by $T_{i,k}^v$, is defined as the subtree of the unfolding tree $F_{i}^v$ induced by the vertices $u \in V(F_{i}^v)$ satisfying
\begin{linenomath*}
\begin{equation*}
\forall w \in V(F_{i}^v)\colon\phi(u)=\phi(w) \Rightarrow  \depth(u)\leq \depth(w)+k.
\end{equation*}
\end{linenomath*}
\end{definition}
Note that for $k\geq i$ the $k$-redundant neighborhood tree is equivalent to the Weisfeiler-Leman unfolding tree. We call the $0$-redundant neighborhood tree simply \emph{neighborhood tree} (NT), see Figure~\ref{fig:exampletrees} for an example.
In the neighborhood tree of a vertex $v$, the edges connecting two vertices with the same shortest-path distance to $v$ are not be represented, e.g., the diagonal edge between the blue and the green vertex of graph $G$ in Figure~\ref{fig:exampletrees}. Hence, we also investigate $1$-redundant neighborhood trees ($1$-NTs), which include such edges. Note that of course, only edges in the same connected component as the vertex can be present in a neighborhood tree.

A $k$-redundant neighborhood tree may contain multiple instances of the same vertex. 
Such duplicates also occur for $k=0$, in which case both vertices necessarily appear on the same level. 
For $k=1$ each edge can only be represented at one level of the tree (but also possibly multiple times).
Therefore, these two choices of $k$ are most interesting to investigate.

Two vertices $u$ and $v$ with $\phi(u)=\phi(v)$ and $\depth(u) = \depth(v)$ are duplicates of a vertex on the same level, and following the definition, the subtrees rooted at $u$ and $v$ are isomorphic. This means, the tree contains duplicate subtrees increasing its size without providing any additional structural information.
Therefore, we propose a \emph{compressed neighborhood tree} (cNT) to avoid this redundancy.
Given two vertices $u$ and $v$ on the same level of a rooted tree $T$, by \emph{merging $u$ onto $v$} we refer to the operation that adds the edge $p(u)v$ to $T$ and deletes the subtree rooted at $u$. When merging two vertices, we assume that no vertices with lower depth can be merged.

\begin{definition}[Compressed $k$-Redundant Neighborhood Tree, cNT]\label{def:cnt}
	The \emph{compressed neighborhood tree} $C_{i,k}^v$ of $v \in V(G)$ is the reduction of $T_{i,k}^v$ obtained by merging vertices $u$ onto vertices $w$ whenever $\phi(u)=\phi(w)$ and $\depth(u)=\depth(w)$. %
\end{definition}
Note that the rooted subtrees of two vertices that are merged are isomorphic leading to a well-defined compressed neighborhood tree independent of the order in which vertices at the same level are merged.
The size and height of a cNT are bounded as follows.

\begin{theorem}\label{thm:size}
  The cNT of a $k$-redundant neighborhood tree rooted at any vertex of a connected graph $G$ has size at most $O(|E(G)|\cdot (k+1))$ and height at most $\diam(G)+k$, where $\diam(G)$ is the diameter of $G$.
\end{theorem}
\begin{proof}
	A cNT $C_{\infty,0}^v$ rooted at vertex $v \in V(G)$ includes all shortest paths from $v$ to any vertex $u$ which have length at most $\diam(G)$. As the parameter $k$ in $k$-NTs allows for up to $k$ repetitions of vertices in subsequent layers, at most $k$ layers are added to the compressed $k$-NT through repetition of the vertex $w \in V(G)$ where $\forall u \in V(G)\colon \depth(u)\leq \depth(w)$. Thereby, the maximum height is bounded by $\diam(G)+k$. Since no vertex can occur multiple times in a single layer due to compression and a vertex can only occur in $1+k$ consecutive layers, we can infer an upper bound on the size of any compressed $k$-NT. Any vertex $u\in V(G)$ can occur at most $k+1$ times together with its incident edges. As $|E(G)|\geq |V(G)|+1$ since $G$ is connected, the size of any $C_{i,k}^v$ is bounded by $O(|E(G)|\cdot (k+1))$.
\end{proof}
The above theorem holds for connected graphs. If the graph $G$ consists of multiple connected components, each can be processed independently and the theorem applies for each connected component. 
Restricting the maximum height reduces the size of the tree further by pruning deeper levels.
The compressed tree $C_{i,k}^v$ can be created directly instead of deriving it from the full tree by merging vertices, making its creation much more efficient.

\begin{algorithm}[tb]
	\caption{Creation of compressed $k$-NT (cNT)}\label{alg:nt_creation}
	\begin{algorithmic}
		\Function{buildNT}{Graph $G$, vertex $w$, height $h$, $k$}
		\State $T \gets \text{cTree(}w\text{)}$ \Comment{initialize cNT with root $w$}
		\State $D\gets$ empty dictionary \Comment{initialize depth dictionary} 
		\State $D(w) \gets 0$
		\For{$i \gets  1, \dots, h$ }
		\State $F \gets$ empty dictionary \Comment{initialize found dictionary}
		\ForEach {$v \in L(T)$}
		\ForEach {$u \in N(\phi(v))$}
		\If{$u$ not in $D$} \Comment{new vertex}
		\State $D(u) \gets i$
		\EndIf
		\If{$D(u) + k \geq i$}
		\If{$u$ not in $F$} \Comment{new on this depth}
			\State add new vertex $c$ to $V(T)$
			\State $\phi(c) \gets u$ 
			\State $F(u) \gets c$
		\EndIf
		\State add new edge $vF(u)$ to $E(T)$
		\EndIf
		\EndFor
		\EndFor
		
		\EndFor
		\Return $T$
		\EndFunction
		
	\end{algorithmic}
\end{algorithm}

\subsection{Creating $k$-Redundant Neighborhood Trees}
\label{sec:creating_nt}
Algorithm~\ref{alg:nt_creation} shows how to generate compressed $k$-redundant neighborhood trees according to Definitions~\ref{def:nt} and~\ref{def:cnt} with a maximum height $h$. The maximum height is naturally bounded by the diameter of $G$. Hence, we set the parameter $h$ to $\min\{h, \diam(G)+k\}$ and assume $h \leq \diam(G)+k$ when the algorithm is applied.
For the creation of a tree the chosen vertex is set as its root and then the neighborhood of this vertex is explored level-wise in a breadth-first search fashion. We build the cNT by keeping track of the first level at which a vertex $v$ was found through $D(v)$, and whether a vertex was already found at the current level through $F(v)$. Then a corresponding vertex has already been created and only the missing edge is inserted. It continues until the maximum height $h$ is reached.
With these techniques Algorithm~\ref{alg:nt_creation} constructs a compressed neighborhood tree of a given vertex in time $O(|E(G)| \cdot (k+1))$.

\subsection{Neighborhood Tree Edit Distance}
We compare pairs of (compressed) neighborhood trees using the SDTED, see Section~\ref{pre:sdted}, to obtain a distance between vertex pairs reflecting the dissimilarity of the neighborhoods for use in the BGM framework. Vertices further away should have a smaller influence on the similarity of two vertices than the vertices themselves or their direct neighbors. Therefore, we introduce a sequence of weights $\bm{\lambda}=(\lambda_1, \lambda_2, \dots)$, where $\lambda_i$ controls the contribution the $i$th level of the neighborhood tree. We multiply the edit costs of the nodes at depth $i$ by $\lambda_i$. We chose $\lambda_i = w^i$ for $0\leq w\leq 1$ and obtain the unweighted case for $w=1$. We chose $w < 1$ to reduce the influence of deeper levels, which correspond to vertices further away.

The SDTED can be computed recursively starting at the root vertices and solving optimal assignments between their children, where the cost for matching two children depends on the SDTED of their subtrees. Schulz et al.~\cite{80_generalizedWLkernel} used this approach with memoization to obtain a running time of $O(nn'h\Delta^3)$ for two Weisfeiler-Leman unfolding trees with $n$ and $n'$ vertices, height $h$ and a maximum degree of $\Delta$.
We improve the running time by introducing techniques for efficient tree matching into the SDTED computation and refining the analysis.
\begin{theorem}\label{thm:runtime}
 The SDTED of two trees with $n$ and $n'$ vertices and maximum degree $\Delta$ can be computed in $O(nn'\Delta)$ time.
\end{theorem}
\begin{proof}
Consider two trees $T=(V,E)$, $T'=(V,E')$ with $n$ and $n'$ vertices, respectively.
We have to solve a linear assignment problem for each pair of vertices $(i,j)\in V\times V'$ with $\depth(i)=\depth(j)$. Let $n_i$ and $n_j$ denote the number of children of $i$ and $j$, respectively. Introducing placeholder vertices for insertion and deletion leads to a squared cost matrix with $n_i+n_j$ rows and columns~\cite{80_generalizedWLkernel}. However, we can find an equivalent solution using the reduction of~Bougleux et al.~\cite{BougleuxGBB20} resulting in an $n_i \times n_j$ cost matrix $\C$. The arising problems can be solved by maximum weight matching algorithms for unbalanced bipartite graphs in time $O(ms + s^2 \log s)$, where $m$ is the number of edges, i.e., non-zero entries in $\C$, and $s$ the smaller vertex set~\cite{Ramshaw2012}. From this we obtain an upper bound of $O(n_i n_j \Delta)$ for our setting, since $m\leq n_i n_j$ and $s=\min\{n_i,n_j\} \leq \Delta$. This leads to a total running time of
\begin{linenomath*}
\begin{equation*}
 \sum_{i\in V} \sum_{j\in V'} O(n_i n_j \Delta) = 
  O\left(\sum_{i\in V} n_i \sum_{j\in V'} n_j \Delta \right) = 
  O(nn'\Delta).
\end{equation*}
\end{linenomath*}
\end{proof}
In the case of integral edit costs bounded by $K$, the costs in all assignment problems are bounded by $C=K(n+n')$. We can use the bipartite matching algorithm by~Goldberg et al.~\cite{GoldbergHKT17} to obtain a total running time of $O(nn'\sqrt{\Delta}\log(\Delta C))$ using the same arguments.

\subsection{Optimization}
\label{sec:caching}
We describe implementation details and techniques to speed up the computation in practice.

The most significant acceleration is achieved by storing and reusing the results of SDTED computations for neighborhood trees and subtrees.
To this end, we use canonical string encodings of trees, which are equal for two trees if and only if the trees are isomorphic.
While no polynomial-time algorithm is known for the graph canonization problem~\cite{DBLP:journals/cacm/GroheS20}, canonical encodings of unlabeled trees can be computed in linear-time by classical algorithms~\cite{Aho1974}, which sort siblings level-wise in a bottom-up fashion. While for unlabeled graphs sorting can be realized in linear-time using bucket sort, for arbitrary label alphabets $O(n \log n)$ time is required to order the tree unambiguously. From this, a string encoding is derived via tree traversal serving as a unique identifier.
In our method, this technique is used for accessing a hash-based data structure storing the SDTED between trees by using tuples of canonical encodings as key.

We implemented the tree canonization procedure such that it generates canonical encodings of all subtrees as a byproduct, which are also cached in the same way. This means that, not only the encodings of subtrees can be reused for computing the encoding of all parents, but also the repeated computation of SDTED between subtrees can be avoided. %
Since Weisfeiler-Leman unfolding trees have a highly repetitive structure, they are expected to greatly benefit from caching.

\subsection{Deriving an Edit Path}
\label{sec:editpath}
Based on the SDTED between their neighborhood trees we obtain a cost matrix for computing a minimum assignment between the vertices of the two graphs.
From the assignment an upper bound for the graph edit distance is then derived yielding our approximation. 

To finally obtain an approximation for the graph edit distance, the same methodology is applied as described in the initial publication using bipartite graph matching~\cite{27_Riesen}, also see Section~\ref{sec:preliminaries:bgm}. The mapping, resulting from the vertex assignment via their neighborhood trees, is seen as a sub-optimal edit path. Vertices, that are mapped to each other incur no cost, if they have the same label, or relabeling cost, if they differ in label. Deletion/insertion cost occur, if a vertex is not mapped to a vertex of the other graph. With this, the rest of the edit path, more precisely the edge edit costs, are inferred from the difference in edges between pairs of mapped vertices. As the overall path derived from this assignment is not guaranteed to be optimal, the cost of the derived path must be greater or equal to the cost of an optimal one. For this reason, our approximation, just like the comparison method, is an upper bound for the graph edit distance.

\subsection{Theoretical Complexity}
We investigate the theoretical complexity of our approach, which is dominated by the time needed for creating the cost matrix $\C$.
\begin{theorem}\label{thm:complexity}
	Given two graphs with $n$ and $n'$ vertices, $m$ and $m'$ edges, and maximum degree $\Delta$.
	The running time for computing the cost matrix $\C$ based on the SDTED for cNTs with constant $k$ is $O(nn'mm'\Delta)$.
\end{theorem}
\begin{proof}
	For each entry of the cost matrix the SDTED between two cNTs is computed. According to Theorem~\ref{thm:size} the size of the cNTs is bounded by $O(m)$ and $O(m')$, respectively. This leads to a running time of $O(mm'\Delta)$ for the SDTED computation with Theorem~\ref{thm:runtime}. Since the cost matrix has $n \cdot n'$ entries, we obtain a total running time of $O(nn'mm'\Delta)$.
\end{proof}
In bounded-degree graphs, where $\Delta$ is constant and the number of edges is a constant multiple of the number of vertices, we obtain a running time of $O((nn')^2)$. 

In Table~\ref{tab:complexity} we compare the theoretical complexity of commonly used approximations for the graph edit distance based on the bipartite graph matching method. We only show the cost to compute the cost matrix $\C$, since the further steps (computing the assignment and deriving the edit path from it) do not differ between the methods. While \emph{BGM} is theoretically the fastest, the substructure that is captured in the cost is not sufficient for good approximation quality in some cases. Note that with increasing radius \emph{Subgraph} is much slower, than the other approaches, since computing the exact graph edit distance between the subgraphs is computationally complex. Compared to the \emph{Walk} method, our approach promises to be more efficient for graphs of bounded degree. We investigate the empirical running time of these methods in Section~\ref{sec:experiments}. 

\section{Experimental Evaluation}
\label{sec:experiments}
We compare our newly proposed technique to state-of-the-art approaches regarding approximation quality and runtime.
Specifically, we address the following research questions:

\begin{itemize}
	\item[\textbf{Q1}] How tight are the upper bounds of our method compared to the state-of-the-art? 
	\item[\textbf{Q2}] How do our bounds perform when taking the trade-off between bound quality and runtime into account?
	\item[\textbf{Q3}] How much of a speed-up is gained by our caching strategy?
\end{itemize}

\subsection{Setup}
\label{subsec:setup}
This section gives an overview of the datasets, the methods used in the experimental comparison and their implementation and configuration.

\subsubsection{Methods and Distance Functions}

We compare our approach to the state-of-the-art bipartite graph matching methods \emph{BGM}~\cite{27_Riesen}, its extensions using bags of walks (denoted by \emph{Walks} or \emph{W})~\cite{bagsofwalks} and subgraphs (denoted by \emph{Subgraph} or \emph{SG})~\cite{10.1007/978-3-319-18224-719}. The competitors employ the same framework as our approach, but offer different trade-offs between quality and running time. Other approximation algorithms are typically orders of magnitude slower.
In addition to the approximation methods we used the state-of-the-art exact GED method \emph{BSS\_GED}~\cite{83_bssged}.

For our method, we compare the assignments under the SDTED using the traditional unfolding trees \emph{WL} and neighborhood trees, where we investigate both $0$-NTs (\emph{NT}) and $1$-NTs ($1$-\emph{NT}), while varying the height parameter for all trees. We employ the optimizations described in the previous section and apply compression the trees generated, including WL unfolding trees. 

The parameter $w$, which determines how the edit costs of distant nodes are weighted down in the SDTED computation was set to $0.5$ for all datasets, which was found to yield good results in preliminary experiments.
Since the exact approach can handle uniform edit costs only, we use these in our experiments.

\subsubsection{Implementation}
We implemented our newly proposed methods, as well as \emph{BGM} and its extensions, in Java. For our competitors methods, we followed the implementations presented by the authors. We used the C\texttt{++} implementation of \emph{BSS\_GED} provided by the authors. All experiments were conducted on an Intel Xeon Gold 6130 machine at 2.1 GHz with 96 GB RAM. We report average execution time over 5 runs. %

\subsubsection{Datasets}
\begin{table}[tb]\centering
	\small
	\caption{Characteristics of the datasets.}
	\label{tab:datasets}
	\setlength{\tabcolsep}{0.18cm}
	\begin{tabular}{lrrrrrr}
		\toprule
		\textbf{Name} & \boldmath\textbf{$\vert$G$\vert$} & \boldmath\textbf{$\overline{\vert V\vert}$} & \boldmath\textbf{$\overline{\vert E \vert}$} & \boldmath\textbf{$\vert L_V\vert$}& \boldmath\textbf{$\vert L_E\vert$}&\textbf{$\overline{diam.}$}\\
		\midrule
		\textit{Letter-med}	&	$2250$ &	$4.67$   & $4.50$ & $-$    & $-$  & $2.44$\\
		\textit{MUTAG}   &	$188$  &	$17.93$ & $19.79$ & $7$    & $4$ & $8.22$\\	
		\textit{PTC\_FM} &	$349$  &	$14.11$ & $14.48$ & $18$  & $4$  & $7.37$\\
		\textit{Protein Com}    & $1000$ &    $10.19$ & $9.19$  & $622$ & $-$  & $4.76$\\
		\textit{MSRC\_9} &  	$221$  &  $40.58$  & $97.94$ & $10$  & $-$  & $7.00$\\
		\textit{NCI1}      & $4110$  &  $29.87$  & $32.30$ & $37$  & $-$  & $13.33$\\
		\bottomrule
	\end{tabular}
\end{table}

We tested all methods on a wide range of real-world datasets from the TUDataset collection~\cite{Datasets} with different characteristics, see Table~\ref{tab:datasets}. 
These datasets are widely used for graph similarity computation and have discrete vertex and edge labels. Attributes, if present, were removed prior to the experiments since not all methods support them. We randomly sampled $1000$ graphs from the dataset \textit{Protein Com}~\cite{Stoecker2019}.

\subsection{Results}
\label{subsec:algorithmsInComparison}

In the following, we report on our experimental results and discuss the different research questions.

\subsubsection{Q1: Bound Quality}
Tight bounds are essential, especially when the exact graph edit distance is small. In this section we compare our new approaches with state-of-the-art bipartite graph matching. We first investigate how many iterations are needed for the different tree structures to achieve the best accuracy.
\begin{figure}[tb]
	\centering
	\includegraphics[width=0.66\linewidth]{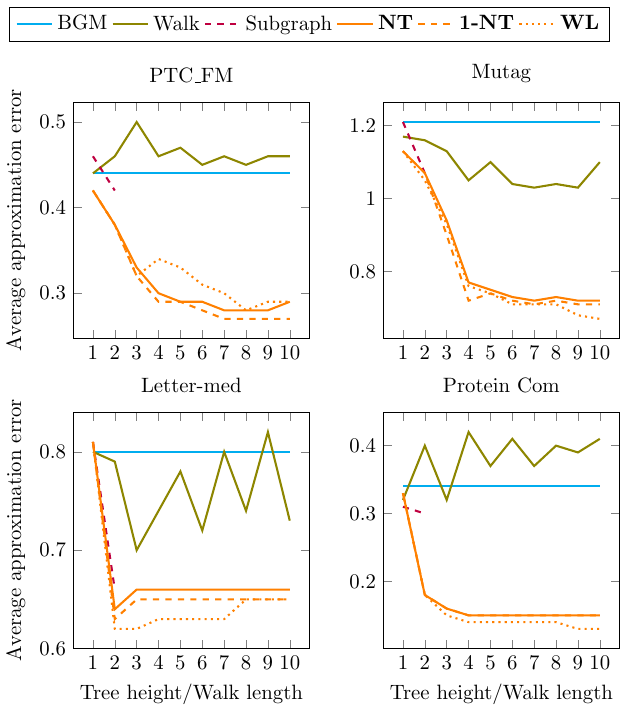}
	\caption{Average relative approximation error of newly proposed (bold) and state-of-the art methods regarding the exact graph edit distance on the different datasets.}
	\label{fig:errorplot}
\end{figure}
For a dataset $\mathcal{D}=\{G_1, G_2, \dots, G_N\}$ we select $n=100$ pairs $(G_i,G_j)$ with $i\neq j$ uniformly at random to obtain a set of pairs $\mathcal{P}=\{(G_1,H_1),(G_2,H_2),\dots,(G_n,H_n)\} \subseteq \mathcal{D}^2$. For each pair we compute the exact graph edit distance $d_i=\ged(G_i,H_i)$ and the result of the approximation method $\hat{d}_i$. We consider the average relative error for this set of pairs, i.e.,
\begin{equation*}
 R=\frac{1}{n}\sum_{i=1}^n \frac{|d_i - \hat{d}_i|}{d_i}.
\end{equation*}
The average relative error of the different methods is shown in Figure~\ref{fig:errorplot}. Since \emph{BGM} does not have any parameter to adjust the tree height, walk length or subgraph radius, the approximation error is constant and it is used as a baseline.
While we gain a huge advantage initially in all datasets, the approximations do not get much better after a certain tree height. An explanation for this is that many graphs are already fully explored by NTs of small height and not many vertices are added in later refinement steps.
The \emph{Walk} method performed sometimes worse and sometimes better than the baseline \emph{BGM}, while always performing worse than our approach. Additionally, the method has unpredictably varying accuracy with increasing walk length, often even leading to degraded accuracy.
The \emph{Subgraph} method could only be tested with a subgraph radius of $1$ and $2$ due to the huge increase in runtime (see Figure~\ref{fig:runtimeplot}). For these parameter settings, the results were comparable to or worse than our method. It can be seen, that $1$-NTs are only slightly better than $0$-NTs in some cases, and the performance of \emph{WL} is also comparable to \emph{NT}.

\begin{figure}[tb]
	\centering
		\includegraphics[width=0.6\linewidth]{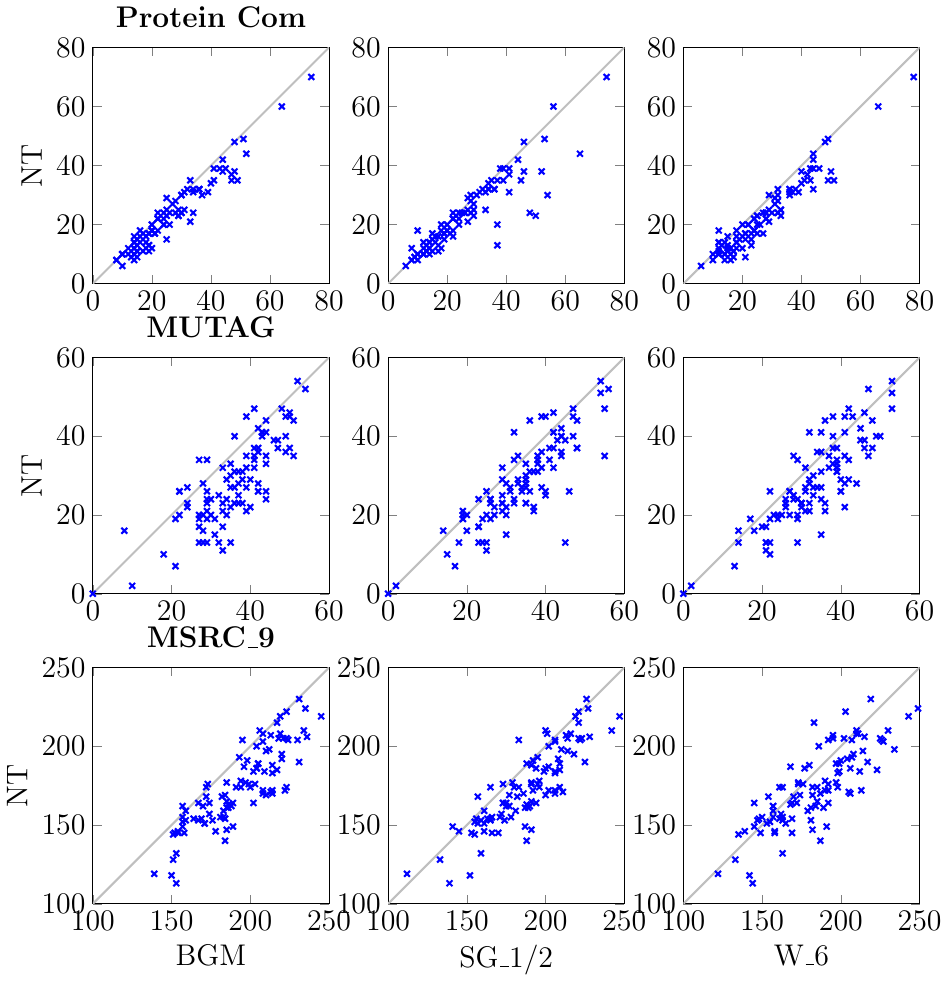}
	\caption{Approximated graph edit distance of \emph{NT\_10} in comparison to the different state-of-the-art methods on the datasets \textit{Protein Com}, \textit{MUTAG and} \textit{MSRC\_9}. The methods are denoted with variant\_parameter.}
	\label{fig:approxscatterMUTAG}
\end{figure}

We compare our approximation \emph{NT} with height $10$ directly to the different state-of-the-art methods on an instance level in Figures~\ref{fig:approxscatterMUTAG} and~\ref{fig:approxscatterNCI}. We denote the comparison methods with variant\_parameter. Points that are below the gray line indicate a better approximation by our approach. Since almost all points are below this line, we can conclude, that our method performs generally better, even on datasets with larger graphs. Especially, when the (approximated) distance is small, \emph{NT} performs much better in many cases. A high accuracy for small distances is very important, for example, in $k$-nearest neighbors search, which is commonly used in data mining tasks such as $k$-nn classification.

\begin{figure}[tb]
	\centering
	\includegraphics[width=0.6\linewidth]{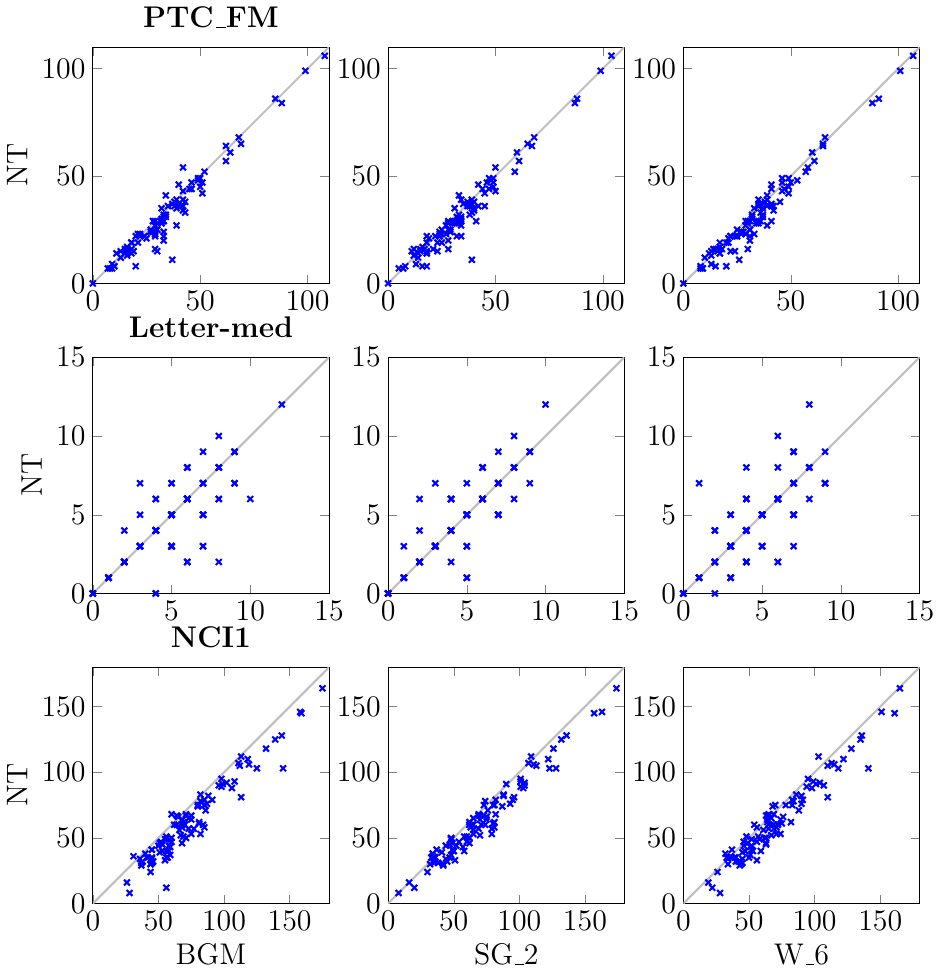}
	\caption{Approximated graph edit distance of \emph{NT\_10} in comparison to the different state-of-the-art methods on the datasets \textit{PTC\_FM}, \textit{Letter-med} and \textit{NCI1}. The methods are denoted with variant\_parameter.}
	\label{fig:approxscatterNCI}
\end{figure}

\subsubsection{Q2: Runtime}
There is typically a trade-off between runtime and accuracy. In this section we evaluate our proposed approaches in terms of runtime. 
Figure~\ref{fig:runtimeplot} shows the average runtime for computing the distance between two graphs using the different methods. Exact computation is much slower than most approximations and  shown as a baseline (for the datasets \textit{MSRC\_9} and \textit{NCI1} no values are given, since the graphs are already too large for this method).
Since \emph{BGM} only computes one larger assignment (and one small assignment for each vertex pair), we can expect it to be faster than our newly proposed methods.
In practice we see that the runtime difference between \emph{BGM} and \emph{NT}/\emph{$1$-NT} is quite small, while our \emph{WL} variant is much slower.
Since the unfolding trees in the \emph{WL} variant contain many redundant vertices and grow very fast in size with increasing height, this is in accordance with our expectation.
For \textit{MSRC\_9} and \textit{NCI1} the \emph{WL} variant is slightly faster than \emph{NT}/\emph{$1$-NT} for lower heights $h$. This can be explained by the impact of our caching strategy, since unfolding trees contain many isomorphic subtrees allowing to skip SDTED computation whenever they reoccur.
The plateauing of the runtime for \emph{NT} and \emph{$1$-NT} can be explained by the bounded size of the neighborhood trees. After every edge in the graph is explored, the tree will not grow in size anymore. The \emph{Subgraph} method becomes very slow when increasing the radius of the subgraphs, to the point where exact GED computation is faster than the approximate method. For this reason, we did not test subgraph radii of $3$ and larger. The runtime of \emph{Walk} varies a lot between the datasets relative to the other methods. Seemingly, \emph{Walk} is fastest when the number of distinct labels is large, due to the small number of matching node pairs in the method's product graph. While \emph{Walk} was the fastest method other than \emph{BGM} for the Protein Com dataset, it is important to note, that it performed worse than the baseline \emph{BGM}.
Interestingly, \emph{$1$-NT} performs much worse on dataset \textit{MSRC\_9} than \emph{NT}, while on the other datasets, there was not much of a difference between them. This could be due to the higher average vertex degree in this dataset, since \emph{$1$-NT} has some redundancy in the tree (by allowing vertices to appear in two consecutive levels). This, and the fact that not much accuracy is gained from using \emph{$1$-NT} over \emph{NT}, suggests that $0$-NTs are sufficient for representing the neighborhood of a vertex accurately and  efficiently.

\begin{figure}[tb]
	\centering
	\includegraphics[width=0.77\linewidth]{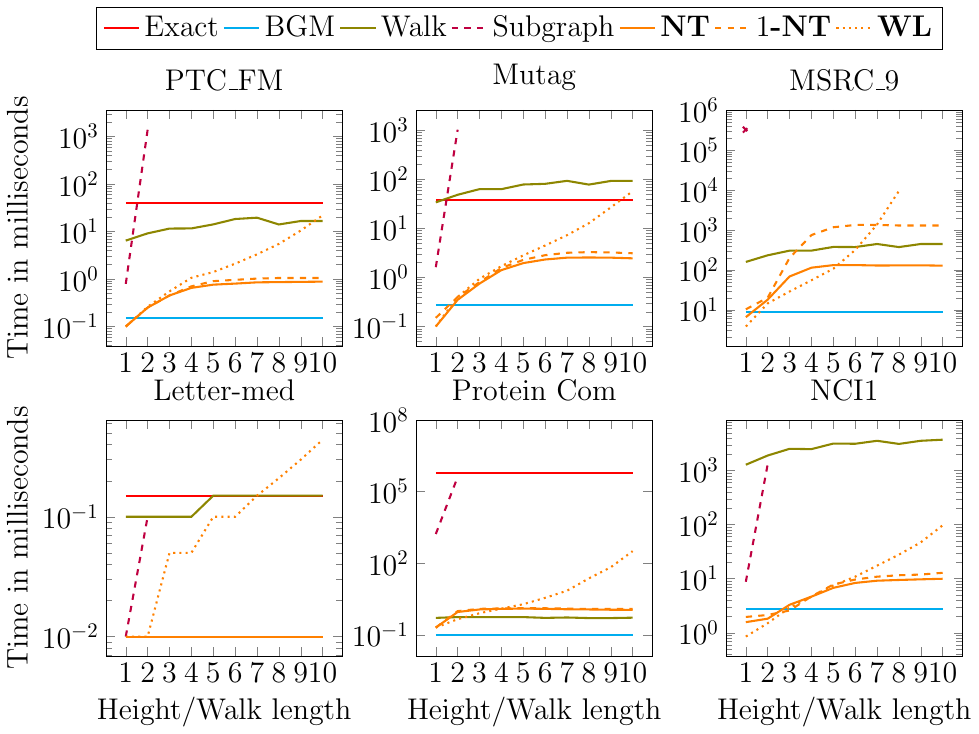}
	\caption{Average runtime for computing the distance between two graphs (proposed in bold).}
	\label{fig:runtimeplot}
\end{figure}

\subsubsection{Q3: Caching}
\label{sec_caching}
In our experiments we evaluate the speed-up in runtime gained using our caching strategy for the different tree variants. For a fair comparison, we do not store the cached values for the whole dataset, but only during the pairwise approximation. If the values were kept over the computation of all or multiple approximations additional speed-up at the cost of memory may be gained whenever subtree pairs reoccur for different graph pairs.

\begin{figure}[tb]
	\centering
	\includegraphics[width=0.66\linewidth]{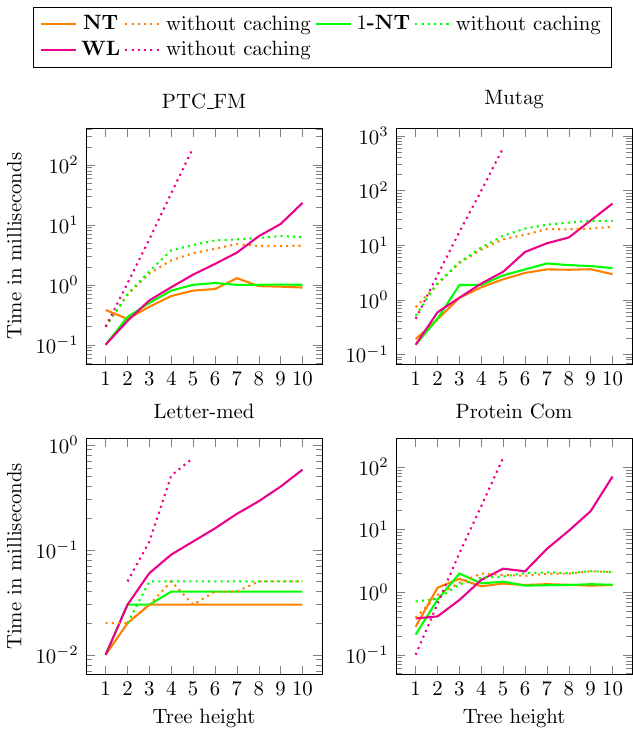}
	\caption{Average runtime in milliseconds for computing the distance of two graphs using our methods with and without our caching.}
	\label{fig:cachetimeplot}
\end{figure}
Figure~\ref{fig:cachetimeplot} shows the average runtime for computing our approximations of the graph edit distance with and without using our caching strategy and varying tree height. Since the runtime for the uncached \emph{WL} increased very quickly with larger height, only values up to a tree height of $5$ are given. On almost all datasets we can see a huge runtime benefit for all methods. Because the size of the \emph{WL} trees is not limited by the graph size and quickly grows, this method benefits the most from caching. The other two methods benefit more on datasets with higher average degree and larger diameter. One reason for this might be the higher number of redundant computations occurring in these graphs, which can be avoided by our caching strategy.

\subsubsection{Discussion}
The experiments showed that, while our approach provides tighter bounds for the graph edit distance, it is only marginally slower than state-of-the-art bipartite graph matching. Especially for graphs with small diameter, a small tree height is often sufficient. We observed, that in contrast to bags of walks, where large walk lengths can decrease the method's accuracy due to the ``tottering phenomenon''~\cite{bagsofwalks}, the accuracy of our method increases with larger tree height. This indicates, that a good choice for our height parameter $h$ can be easily and intuitively found by a user, as it trades accuracy for runtime. Our experiments have shown that not limiting $h$ still results in a fast method, as the increase in computational complexity arising from the parameter $h$ is bounded by the diameter of the graphs. 

\section{Conclusions}
\label{sec:conclusion}
We introduced a new approximation for the graph edit distance based on bipartite graph matching. For the vertex assignment, neighborhood trees and their edit distance are used.
Thereby, our approach takes complex structural information around the vertices into account without performing expensive exact GED computations for local subgraphs. As our trees can be limited in height, an individual trade-off between runtime and accuracy can be achieved. 
We analyzed the running time of the SDTED computation and improved it to $O(n^2\Delta)$ for two trees with $n$ vertices and maximum degree~$\Delta$.
Moreover, we proposed to accelerate the computation by using compact compressed tree representations and tree canonization to avoid redundant computations.
In our experimental evaluation we showed that the proposed method is often only slightly slower than standard bipartite graph matching, while providing great benefits regarding approximation accuracy.

In future work we will investigate other practical applications of our method, such as  database search, where storing the computed SDTEDs globally seems promising. Furthermore, our method can be applied to molecular graphs for structure-based virtual screening of ligands~\cite{doi:10.1021/acs.jcim.8b00820}. The improved approximation quality with only small overhead is promising especially when working with larger molecules such as macrocyclic compounds~\cite{SENGUPTA20201851}. 
The Weisfeiler-Leman algorithm is fundamental for graph learning methods such as graph kernels and graph neural networks~\cite{DBLP:conf/ijcai/0001FK21}. Various applications in this field could potentially benefit from replacing Weisfeiler-Leman unfolding trees by more compact neighborhood trees.

\section*{Acknowledgements}
This work was supported by the Vienna Science and Technology Fund (WWTF) [10.47379/VRG19009]. 
\bibliography{lit}

\end{document}